\documentclass[12pt]{article}
\usepackage[english]{babel}%required allways
\usepackage[onehalfspacing]{setspace}
\usepackage{geometry}
\usepackage{indentfirst}
\usepackage{amsmath, amsfonts,  eurosym, textcomp, graphicx, amsthm, amssymb, url,  xcolor,tikz, subfig}
 \usepackage{titlesec}
\usepackage{mathrsfs}
\usepackage{etoolbox}
\usepackage{amsmath, mathtools}
\usetikzlibrary{arrows,patterns}
\usepackage{a4wide}

\usepackage[sort,comma,authoryear,round]{natbib}
 \titleformat{\section}
  {\normalfont\fontsize{14pt}{14pt}\bfseries }{\thesection}{1em}{}
\titleformat{\subsection}
  {\normalfont\fontsize{12pt}{12pt}\bfseries }{\thesubsection}{1em}{}

\newtheorem{prop}{Proposition}
\newtheorem{df}{Definition}
\newtheorem{tm}{Theorem}
\newtheorem{lm}{Lemma}
\newtheorem{as}{Assumption}
\newtheorem{ex}{Example}

\usepackage{bbold}

\DeclareMathSymbol{\N}{\mathbin}{AMSb}{"4E}
\DeclareMathSymbol{\Z}{\mathbin}{AMSb}{"5A}
\DeclareMathSymbol{\R}{\mathbin}{AMSb}{"52}
\DeclareMathSymbol{\Q}{\mathbin}{AMSb}{"51}
\DeclareMathSymbol{\I}{\mathbin}{AMSb}{"49}
\DeclareMathSymbol{\C}{\mathbin}{AMSb}{"43}
\DeclareMathSymbol{\F}{\mathbin}{AMSb}{"46}

\usepackage[breaklinks,colorlinks, citecolor=blue, linkcolor=black, urlcolor=blue,pdfsubject={LaTeX}]{hyperref}

\begin{document}
\pagenumbering{gobble}

\title{Reduced-Form Allocations with Complementarity:  A 2-Person Case\thanks{I thank Tilman B\"orgers, Eric van Damme, Monique Laurent, Debasis Mishra and Zaifu Yang for their helpful comments. I am grateful to Rakesh Vohra for his numerous comments which greatly improved this paper.} }
\author{Xu
Lang\thanks{%
Department of Economics, Southwestern University of Finance and Economics, Chengdu, China; langxu@swufe.edu.cn.} }
\date{  February 22, 2022}
\maketitle

\begin{abstract}
We investigate the implementation of reduced-form allocation probabilities in a two-person bargaining problem without side payments, where the agents have to select one alternative from a finite set of social alternatives. We provide a  necessary and sufficient condition for the implementability.   We find that the implementability condition in bargaining  has some new feature compared to Border's theorem. Our results have applications in compromise problems and package exchange problems where the agents barter  indivisible objects  and the agents   value the objects as complements.
\bigskip


\end{abstract}

\pagenumbering{arabic}

\newpage

\section{Introduction}

The bargaining problem of \citet{NA50} considers a social situation in which two agents can barter some indivisible goods with lotteries but there is no money to facilitate exchange. \citet{MY79, MY84} extends this bargaining problem to an incomplete information environment, in which each player has only statistical information about the other's preferences. Which lottery
over the alternatives will  the players implement as a fair compromise at the interim stage?  To analyze feasible allocations in this bargaining problem, one needs to resolve two kinds of basic constraints: First, an interim allocation probability has to satisfy  the usual incentive compatibility. Second, such an allocation probability has to be reduced-form implementable: there exists an ex post feasible allocation probability generating this interim allocation probability.  In order to find the set of all interim feasible allocations, it is important to characterize  implementable reduced-form allocations.

The implementability of reduced-form allocations has been well studied in  the classic one-dimensional auction design problem \`a la \cite{MY81}.   \cite{MR84} first study this question and obtain a partial solution.   \cite{MA84} poses a conjecture on the set of implementability inequalities, and \cite{BO91, BO07} obtains a complete characterization by a geometric approach. Alternative approaches have  been  developed recently. \cite{CKM13} propose a network flow approach and generalize Border's theorem to multi-unit auctions with polymatroidal constraint structures.\footnote{\cite{AFH19} develop a  polymatroidal decomposition approach and obtain a similar characterization.} \cite{ HR15} and \cite{KMS21} obtain a related characterization by the theory of majorization inequalities.  \cite{GK11, GK16, GK20} consider a social choice problem environment and further develop the geometric approach.  \cite{LY19} provide a conic characterization of reduced form auctions. \cite{Zh21} extended the network flow analysis to multiple  objects with multiple units assuming the paramodular constraints on quotas.  Due to its analytical tractability, the reduced-form characterization has been a powerful tool to study various types of mechanism design problems.\footnote{See for example \cite{AM01},  \cite{HS11},  \cite{MIR12},  \cite{PV14}, \cite{MZ17}, and \cite{GMSZ20}. }

A fundamental assumption in the auction problem of \citet{MY81} is that  each of the buyers is interested in his own chance of winning, i.e., there is no allocative externalities among the buyers. On the other hand, in a 2-person Bayesian bargaining or compromise problem with public alternatives \citep{MY79, MY84, BP09},  as well as mechanism design problems with
allocative externalities (e.g., \citealt{JMS96,JMS99}),  each social alternative may influence all players.\footnote{\cite{BP09} observe that the reduced-form implementation problem in a compromise problem differs from that in the classic auction problem.} In these cases, multiple public alternatives determine a multidimensional reduced form probability for each player, which distinguishes these problems from the standard auction problems.

In this paper, we study the implementation problem in Myerson's bargaining problems without side payments. We first translate the implementation problem into a directed multiflow problem in an appropriately defined network. To fully characterize the implementability condition,  we introduce a graph transformation technique  due to  \cite{EV78a} and \cite{ST80}, which reduces the directed multiflow problem into a classic single-commodity flow problem. Then by invoking a version of Hall's theorem, we characterize the necessary and sufficient condition for the implementability.

To illuminate the computational tractability of our results, we study polyhedral properties of the set of feasible reduced forms.\footnote{See \cite{AFH19}, \cite{CDW17}, and \cite{GNR15} for a discussion on the computational complexity of the reduced-form approach. } We characterize the implementability inequalities by the notion of lattice polyhedra due to \cite{GH82}. We compare our result to Border's condition, and show that the implementability condition in a bargaining problem has some new feature compared to Border's condition in a standard auction.

Our  implementation result can be useful to study axiomatic solutions of bargaining with incomplete information. Nash's bargaining solution is the best known solution in 2-person bargaining problems. \cite{HS72} provide a generalization of Nash's solution to bargaining games with incomplete information. \cite{MY79, MY84} and \cite{BP09} consider two-person bargaining and compromise problems with a finite set of alternatives and further incorporates incentive compatibility from a mechanism design perspective. Since the generalized Nash product is a non-linear objective,  the implementability condition could be binding in an optimal solution and required for a characterization. This condition can be further used to obtain a reduced-form characterization of the sets of incentive feasible and incentive efficient mechanisms (e.g., \citealt{HM83}).

While we focus on a bargaining model without side payments, the implementation result also
applies to other preference domains with the same reduced-form structures. In particular, it applies to multidimensional mechanism design problems with quasi-linear utilities, i.e., multi-item auctions with sub (super)-modular valuations and public good problems with multiple alternatives, when the number of privately informed players is restricted to two.

In multi-item auction problems, packaging decisions are important whenever a user's value for
a package is different from the sum of the values of the separate items, i.e., spectrum licenses can be either substitutes or complements (e.g., \citealt{AM02, MGM07, SY06}).\footnote{See also \cite{SAM74} for a discussion on complementarity.} Compared to auction design with additive valuations, designing package auctions and exchanges is more challenging \citep{MGM07}. The reduced-form implementation in such a problem also differs from the classic analysis. When every agent has an additive valuation function for the items,
Border's theorem can be applied separately to each item (e.g., \citealt{AM01, MIR12, CDW17}). However, when the agents' valuation functions are non-additive, Border's condition
cannot be applied separately to each item.  The multidimensionality of reduced-form allocation rules distinguishes the package auction problems from the one-dimensional auction problems.  Then the implementation result in this paper will be relevant for the analysis. 

It is worth noting that \cite{GK11,GK16, GK20}  also study an implementation problem in social choice environments.  They characterize the support function for the probability simplex constraint and obtain an implementability conditions for reduced form values and allocation probabilities. Their support function characterization  applies to general social choice problems  and covers our bargaining problems.  On the other hand, while their condition needs to be checked for every  possible system of real-valued weights, the condition in our paper identifies the facet-induing weights and needs to be checked for finitely many inequalities.

In Section 2, we present a package exchange example which generates an implementation problem in our setting. 
In Section 3, we introduce Myerson's bargaining model and the corresponding implementation problem. In Sections 4 and 5, we provide a characterization for the implementability.  In Section 6, we analyze the structures of feasible reduced forms. Section 7 concludes.

\section{A Motivating Example:  Package Exchange}

We provide a simple example to illustrate how to generate an implementation problem in a two person exchange economy with complementary objects and how it differs from \cite{BO91}. Suppose there are two agents, 1 and 2, and two heterogeneous indivisible
items, $A$ and $B$, initially owned by agent 1  and agent 2 respectively. Agents can barter the items with lotteries but there is no money. For each
package $S \subseteq \{A,B\}$,  agent $i$ has a private valuation $v_i(S)\in \R_+$. Assume $v_i(\emptyset)=0$. Define the set $D$ of social alternatives by all possible partitions of the items between the agents 1 and 2: $D=\{(AB,\emptyset), (A,B), (B,A), (\emptyset, AB)\}$. We consider two classes of ex post feasibility constraints with lotteries.

Suppose first that as in \cite{BO91}, two items are allocated separately and we use independent lotteries $(q^A,q^B)$ for $A$ and $B$, which require the following
feasibility constraint for each item,
\begin{equation}
q_1^S+q_2^S=1, q^S\geq 0, \,\,\text{for}\,\, S=A,B.
\end{equation}
Multiplying the above equations for $q^A$ and $q^B$, we obtain the probabilities of each social alternative
being chosen by
\begin{equation}
q_1^Aq_1^B+q_1^Aq_2^B+q_1^Bq_2^A+q_2^Aq_2^B=1, q\geq 0, \,\,\text{for}\,\, S=A,B.
\end{equation}
Suppose next that we conduct a correlated lottery $q$ over social alternatives, which allows the possibility of bundling the items,
\begin{equation}
q\geq 0,  \sum_{d\in D} q^d=1.
\end{equation}
Under alternative valuation assumptions, the difference between correlated and independent
lotteries is immediate: With additive valuations (i.e., $v_i(A)+v_i(B)=v_i(AB)$), each agent $i$ will be interested in $(q_i^A, q_i^B)$ only, and hence using independent lotteries $(q^A, q^B)$ is without loss.\footnote{\cite{AM01} characterize the optimal auction where the agents' valuations for the objects are additive.} With complementary valuations (i.e., $v_i(A)+v_i(B)<v_i(AB)$), agent $ i$ will be interested in $q^d$ in general, i.e., the probabilities that agent $ i$ obtains different packages. In this case, restricting to independent lotteries rules out many feasible allocations. 

Therefore, in the presence of complementarity, it is necessary to define the feasibility constraint (and hence the implementation problem) over the entire set of social alternatives. The example shows that even for private good allocations, the assumption on valuation functions is important for the formulation of an reduced-form implementation problem. An auction problem with complementarity has the same reduced-form structure as a bargaining problem with public alternatives.

\section{The Model}

We consider a simple two-person \textit{Bayesian bargaining problem} without side payments in \cite{MY84}. There are two players, 1 and 2, and a finite set   of social
alternatives $K=\{k_0,k_1,...,k_n\}$.
 For each player $i$, there is   a finite set $T_i$  of possible types for player $i$.   Let $u_i: K\times T_i\to \R$ denote a utility function of player $i$, measured in a von Neumann-Morgenstern utility scale. We assume that at least one of the players has a binary type space:
   \begin{as} $min \{|T_1|, |T_2|\}= 2$.  
 \end{as}
 
The  model fits many stylized  situations in bargaining with risk averse agents, compromise models and package exchange models: We provide three examples in the following.

 \begin{ex}  (Private risk attitude). One of the results most frequently quoted in the bargaining literature is that increasing risk aversion may hurt a player in the bargaining outcome (e.g., \citealt{KRS81}). Suppose players' risk attitudes are private information.  For each alternative $k\in K$, let $w_{ik}$ be the  monetary payoff that player $i$ receives from alternative $k$ and assume the payoff matrix $w=(w_{ik})$ is common knowledge. Each player can be either risk neutral ($t^a_i$) or risk averse ($t^b_i$), where
 \begin{equation}
u_i(k, t^a_i)= w_{ik},   \,\,{and}\,\, u_i(k,  t^b_i)=\sqrt{ w_{ik}}.\notag
 \end{equation}
\end{ex}

 \begin{ex}  (Private compromise payoff). Consider a compromise problem in \cite{BP09}. There are three alternatives, $k_1, k_2$ and $k_0$, and two players 1 and 2 with opposite preferences: $k_1\succ_1 k_0\succ_1 k_2$ and $k_2\succ_2 k_0\succ_2 k_1$.\footnote{The compromise problem of \cite{BP09} assumes that there is no disagreement outcome. When $k_0\in K$ is selected as the disagreement  outcome, the model can be interpreted as a bargaining problem in \cite{MY84}.}   We normalize $u_i(k_i)=1$ and $u_i(k_j)=0$. Assume that players are risk neutral and each player has private information about her payoff on the compromise alternative $k_0$, i.e., $u_i(k_0)=t_i$. Each player's type can be either strong or weak, with a strong type receiving a higher payoff from $k_0$.
 \end{ex}
 
  \begin{ex} (Private package valuation). Consider the package exchange example in Section 2, where agents can exchange the two items they own. Our model allows for domains that cover ordinal and cardinal preferences. A type $t_i$ may represent a valuation vector $t_i=(v_i^A, v_i^B, v_i^{AB})$. Alternatively,  a type  $t_i$ may represent an ordinal preference over all possible packages. For example, for each agent i, let $T_i=\{\succsim_i,\succsim_i'\}$ where
   \begin{equation}
AB\succsim_i A\succsim_i B  \,\,{and}\,\, AB\succsim_i' B\succsim_i' A.\notag
 \end{equation}
An agent's preference over lotteries can be then defined by first order stochastically dominance relation.
 \end{ex}

We let $T=T_1\times T_2$ denote the set of all possible  type profiles and let $\lambda $ be a probability measure on $T$, as the common prior of the players. We assume $\lambda(t)>0 $ for all $t\in T$.  For each player $i$, let  $\lambda_i$ denote the marginal probabilities induced by $\lambda$.  We assume $\lambda$ is statistically independent.

 An  (ex post) feasible allocation rule assigns to each type profile a lottery over social alternatives. Formally,  an ex post allocation rule  $q: K\times T\to \R$ is \emph{feasible} if
\begin{equation}\label{eq:1}
q\geq 0 \phantom{0}  \text{and} \phantom{0}  \sum_{k\in K} q^k(t)= 1, \phantom{0}  \text{for all} \phantom{0} t\in T.
\end{equation}
The component $q^k(t)$ denotes the probability that  alternative $k\in K$ is chosen given type profile $t \in T$.

An ex post feasible allocation rule $q$ induces an interim allocation rule  $Q=(Q_1,Q_2)$, where $Q_i:K\times T_i  \to \R$   denotes   player $i$'s interim expected allocation probabilities  given his type.   For each  $i=1,2$, $t_i\in T_i$, and $k\in K$,
\begin{equation}\label{eq:2}
Q_i^k (t_i):=\sum_{t_{-i}\in T_{-i}}q^k (t) \lambda_{-i}(t_{-i}).
\end{equation}
(We use the notations $T_{-i},t_{-i},\lambda_{-i}$ for the player other than $i$.) We then say $Q$ is the \emph{reduced form} of $q$ and $q$ \emph{implements} $Q$.

Conversely, one could begin with an arbitrary interim allocation rule and ask whether it can be implemented by an ex post feasible allocation rule or not.

\begin{df}  An interim allocation rule $Q$ is {\it implementable}, if   there exists an ex post allocation rule $q$ such that $(q,Q)$ satisfies both \eqref{eq:1} and \eqref{eq:2}.
\end{df}

Note that  $(\{1,2\}, K, T, \lambda)$ fully describes an implementation environment. We state an implementation problem $\mathcal{I}=(\{1,2\}, K, T, \lambda,Q)$ as follows:  Pick any interim allocation rule $Q$ and determine whether it is implementable or not.

By inspection, we obtain the following necessary conditions on $Q$ for the implementability. First, if  $Q$  is implementable, then
\begin{equation}
  \sum_{k\in K} Q^{k}_i(t_i)=1, \,\,\text{for all}\,\,t_i\in T_i, i=1,2.\label{eq:3}
    \end{equation}
Denote by $K_*$ the set of social alternatives other than $k_0$. From condition \eqref{eq:3}, it follows that  each $Q_i^{k_0}$ is a slack variable determined by  $(Q_i^k)_{k\in K_*}$. Below   the slack variable $Q^{k_0}$ will be explicitly taken into our analysis but we state our characterization result with variables $(Q^k)_{k\in K_*}$ only.

\section{ Characterization}

\begin{lm} If $Q$ is implementable, then 
\begin{align}
&\sum_{t_1\in T_1} Q_1^k(t_1)\lambda_1(t_1)-\sum_{t_2\in T_2} Q_2^k(t_2)\lambda_2(t_2)=0, \,\,\text{for all}\,\,k\in K_*, \label{eq:4}
 \end{align}
and
 \begin{align}
&Q^k_i\geq 0, \,\,\text{for all}\,\, k\in K_*, i=1,2.\label{eq:5}
 \end{align}

\end{lm}
Condition \eqref{eq:4} requires that at the ex ante stage, the two players must have consistent beliefs on the probability allocated to each alternative.  Note that \eqref{eq:4} and \eqref{eq:5} define a pointed cone.\footnote{The system of equations \eqref{eq:4} defines a  linear subspace $L$ of $\R^{(|T_1|+|T_2|)\times|K_*|}$: For $Q',Q{''}\in L$, it implies $\alpha Q'+\beta Q{''}\in L$ for all $\alpha, \beta\in \R$.}  We call conditions \eqref{eq:4}-\eqref{eq:5} the conic condition.  

\begin{lm} If $Q$ is implementable, then 
\begin{equation}\label{eq:cut2}
\sum_{k\in G}[\sum_{t_1\in E_1} Q^k_1(t_1)\lambda_1(t_1)-\sum_{t_2\in   E_2} Q^k_2(t_2)\lambda_2(t_2)]   \leq  \lambda(E_1\times E_2^c),
 \end{equation}
for all $G\subseteq K_*$,  $ E_1\subseteq T_1 $, and $E_2\subseteq T_2$.

\end{lm}
Condition \eqref{eq:cut2} can be interpreted as saying that  for any set of types $E_1\times E_2$, the difference in the players' interim beliefs on the allocation probability of any subset  $G\subseteq K_*$ of alternatives  cannot be too distinct. In contrast to Border's condition where coefficients $\{0,1\}$ fully describe the implementability inequalities,  condition \eqref{eq:cut2} requires coefficients $\{-1,0,+1\}$. To interpret this result, notice that in Border's condition, selling to one buyer with a higher  expected probability  tightens the probability budget for another buyer.   The buyers are competing for the  expected probabilities of winning. In a bargaining model, selecting alternative $k$ with a higher  expected probability for one player, however,   relaxes the probability budget of alternative $k$  for the other player. The players have common interests at each alternative.

The following Theorem 1, which is the main result of the paper, shows that the necessary conditions in Lemmas 1 an 2 provide a complete description of the implementability condition.

\begin{tm}  $Q$ is implementable if and only if conditions \eqref{eq:4}, \eqref{eq:5}, and \eqref{eq:cut2} hold.
\end{tm}

The implementability condition in Theorem 1 generalizes a well-known condition of \cite{ST65} and \cite{GKRS91}, which gives a necessary and sufficient condition for the existence of measures with given marginals. To see this, suppose $K=\{k_0,k_1\}$ and define  for $i=1,2$,
\begin{equation}
\nu_i(E_i):=\sum_{t_i\in E_i}Q_i^{k_1}(t_i)\lambda_i(t_i)\end{equation}
 in the implementability condition. Then we obtain Strassen's condition
\begin{equation}
\nu_1(E_1)\leq \nu_2(E_2)+ \lambda(E_1\times E_2^c),
\end{equation}
for all $ E_1\subseteq T_1 $ and $E_2\subseteq T_2$.

\cite{GGKMS13} observe that for standard auctions with two buyers (i.e., $|K|=2$), Border's condition is equivalent to Strassen's condition. We note that our characterization differs from the network flow characterization of \cite{CKM13} in several aspects. First, \cite{CKM13} translate the implementation problem in one-dimensional auctions  into a directed single-commodity polymatroidal flow problem. To derive the implementability condition, they invoke a feasible flow characterization.  Their  characterization for single-commodity flow problems hence is not applicable to a multiflow problem. Second, while   \cite{CKM13} show that the implementation system in their problem is the intersection of two polymatroids (see \citealt{ED70}) and forms a totally dual integral (TDI) system,  the multiflow system we construct  is not TDI in general.  Finally, while \cite{CKM13} construct a network flow problem with capacitated arcs, the network formulation in our model has uncapacitated arcs and capacitated nodes.

\section{Proof of Theorem 1}
In this section, we provide a sketch of proof of Theorem 1. We first formulate the implementation problem as a multiflow problem. We then apply a graph transformation  which reduces the multiflow problem to a classic flow problem. Finally, using a feasible flow theorem, we show that the  conditions \eqref{eq:4}-\eqref{eq:cut2} is necessary and sufficient for the implementability.

We first introduce a multiflow problem. Let $G=(V,A)$ be a directed network and let $J$ be a set of commodities. For each commodity $k\in J$, we partition the nodes $V$ into $S^k, \bar{S}^k$, and $N^k$, which denote the sources, sinks, and transit nodes.  Notice that the model allows multiple sources and sinks for each commodity. For each node $v$, let $d^k(v)$ be the net demand of commodity $k$.  Hence  $d^k(v)<0$ if $v\in S^k$,  $d^k(v)>0$ if $v\in \bar{S}^k$, and $d^k(v)=0$ if $v\in N^k$. Assume $\sum_{v\in V} d^k(v)=0$. For each arc $a\in A$, let $c(a)$ be the capacity. Define a multiflow system by
\begin{align}
\label{eq:T2a}& M_G f^k=d^k, \,\,   \text{for all}\,\, k,\\
\label{eq:T2b}& \sum_{k} f^k+s=c, \,\,  \\
\label{eq:T2c} &f\geq 0,  s\geq 0,
\end{align}
where  $M_G$ is the node-arc incidence matrix of $G$, $f=(f^k)$ is a multiflow variable, and $s$ is a slack variable associated with capacity constraints.  A multiflow problem $P=(G, (S^k, \bar{S}^k, d^k)_{k\in J}, c) $ is to determine whether a feasible multiflow satisfying \eqref{eq:T2a}-\eqref{eq:T2c} exists or not.

We now formulate $ \mathcal{I} $ as a multiflow problem $P_{\mathcal{I}}$ of the form \eqref{eq:T2a}-\eqref{eq:T2c}. Define graph $G$ by $(T_1\cup T_2, T)$, where $T_1$ and $T_2$ consist of all nodes and $T$ consists of all arcs, from each $t_1\in T_1$ to each $t_2\in T_2$.  
Define the set $J$ of commodities by $K_*$. For each commodity $k$,  we define  the sources $S^k$ by  $T_1$  and the sinks  $\bar{S}^k$ by $T_2$.  For each node $t_1$ (and $t_2$), and commodity $k$, we define the net demand $- Q_1^k(t_1)\lambda_1(t_1) $ (and $Q_2^k(t_2)\lambda_2(t_2) $).     For each arc $(t_1,t_2)\in T$, define the capacity $\lambda(t)$.  For each arc $(t_1,t_2)\in T$, define the multiflow variables $f^k(t_1,t_2)=q^k(t)\lambda(t)$, $k\in K_*$, and the slack variable $s(t_1,t_2)=q^{k_0}(t)\lambda(t)$. We have the following equivalence between the implementation problem and the multiflow problem.
\begin{lm}
 $\mathcal{I}$ has a feasible solution if and only if $P_{\mathcal{I}}$ has  a feasible multiflow.
\end{lm}

\textbf{Graph transformation procedure.}    \cite{EV78a} and  \cite{ST80} introduced a  class of directed graphs called suspension graphs, and showed that a multiflow problem on a suspension graph can be transformed into a classic flow problem.  For a directed graph $G$, we refer to a {\it cycle} as a simple cycle in which (a) the  arcs have arbitrary directions and (b) the only repeated vertices are the first and last vertices. We say $G$ is {\it connected} if there is a path between every pair of vertices.  $G$ is {\it 2-connected} if it is connected, and if every pair of arcs is contained in at least one cycle. 

\begin{df}  A directed graph $ G$ is a {\it suspension}, if it is 2-connected and there exists a node $v_*$ such that after deleting  $v_*$  and the arcs incident to it, the graph does not contain any cycle. \end{df}

The following lemmas show that  if $G$ is a suspension  with respect to some node $v_*$, then problem $P$ has unimodular constraint matrix and it can be transformed into a classic network flow problem $P(v_*)$.\footnote{A matrix is unimodular if every basis has determinant $1$ or $-1$.  } Indeed,   each variable in $P(v_*)$ appears in exactly two constraints with opposite signs, and hence the constraint matrix of $P(v_*)$ is the node-arc incidence matrix of a directed graph, which we define as the transformed graph $G(v_*)$. Also, in the network $G(v_*)$, the supply and demand at each node  is defined by the original system, i.e., each row in $P(v_*)$ is a row in $P$. Notice that $P(v_*)$  is a linear program in standard form, and hence a network flow problem with uncapacitated arcs.

\begin{lm}\citep{ST80} Let $P=(G, (S^k, \bar{S}^k, d^k)_{k\in J}, c)$ be a multiflow problem. If $G=(V,A)$ is a suspension, then the constraint matrix of $P$ is unimodular.
\end{lm}

\begin{lm}\citep{EV78a,ST80} Let $P=(G, (S^k, \bar{S}^k, d^k)_{k\in J}, c)$ be a multiflow problem where $G=(V,A)$ is a suspension with some node $v_*$. Define a network flow problem $P(v_*)$ by
\begin{align}
& f^k(\delta^{in}(v))-f^k(\delta^{out}(v))=d^k(v),   \,\,\text{for all}\,\,k, v\in V\setminus \{v_*\}, \notag \\
&s(\delta^{in}(v))-s(\delta^{out}(v))=d^s(v),   \,\,\text{for all}\,\,v\in V\setminus \{v_*\}, \notag\\
& \sum_{k} f^k(a)+ s(a)=c(a), \,\,\,\,\,\text{for all}\,\, a\in A(v_*),\notag \\
& f\geq 0,  s\geq 0,\notag
\end{align}
where $A(v_*)=\{(x,y)\in A: x\,\,{or}\, y=v_*\}$ and $d^s(v)=c(\delta^{in}(v))-c(\delta^{out}(v))-\sum_{k}d^k(v)$. Then $P$ has a feasible solution if and only if  $P(v_*)$ has a feasible solution.
\end{lm}
 
\textbf{Transformation for  $\mathcal{I}$.}     We invoke this procedure to our implementation problem. First note that $(T_1\cup T_2, T)$ with $T_1=\{t_1^a,t_1^b\}$ is 2-connected, and that after deleting $t_1^a$ and the arcs incident to it,  the graph does not contain any cycle. Hence  $(T_1\cup T_2, T)$ is a suspension. By Lemma 5, we set $v_*=t_1^a$ and obtain the transformed problem $P(t_1^a)$ of $P_{\mathcal{I}}$, which is a classic network flow problem with capacitated nodes and uncapacitated arcs. To complete the proof, we  invoke a version of Hall's theorem to obtain a characterization for the implementability.  The details of proof are left to Appendix A.

\section{Structures of Reduced Forms}

In this section, we investigate the structures of implementable reduced forms  in a bargaining problem and obtain several useful characterizations. We also compare the implementability condition to  Border's condition and illustrate how the  structures of reduced forms in these problems differ.  We will focus on the following two sets of vectors $Q=(Q^k)_{k\in K_*}$, where we denote

1. $\mathcal{Q}$:  the set defined by the conditions \eqref{eq:4}, \eqref{eq:5}, and \eqref{eq:cut2}.

2. $\mathcal{Q}^*$:  the set defined by the conditions \eqref{eq:5} and \eqref{eq:cut2}.

\subsection{Lattice Polyhedron}

\cite{VO11} shows that the set of feasible reduced form auctions is a polymatroid. This implies that  feasible reduced forms that optimize over a given linear objective can be found by the greedy algorithm. \cite{CKM13} and \cite{AFH19} show that this property generalizes to the auction problems with group capacity constraints and matroid constraints.  \cite{AFH19} develop a polymatroidal decomposition approach to show that feasible reduced forms are a polymatroid associated with an expected rank function. They also provide computationally tractable (i.e., in polynomial time in the total number of agents’ types) methods for optimization and implementation of interim allocation rules. In this subsection, we investigate the polyhedral aspect of the implementability condition in the bargaining problem and show that the implementability condition has a richer structure than a polymatroid. 

Let  $\mathcal{A}=2^{T_1}\times 2^{T_2}\times 2^{K_*}$. Define a partially ordered set $(\mathcal{A},\preceq)$ by set inclusion and set intersection for each coordinate, i.e., for any $A^1,A^2\in \mathcal{A}$ with $ A^l=(E_{1}^l,E_{2}^l,G^l)$, $A^2 \preceq A^1$ if and only if $E_{1}^2 \subseteq E_{1}^1$, $E_{2}^2\subseteq E_{2}^1$,  and $G^2 \subseteq G^1$. Then $(\mathcal{A},\preceq)$ defines a lattice $(\mathcal{A},\preceq,\wedge, \vee)$ with lattice operations $\wedge, \vee$: for any $A^1,A^2\in \mathcal{A}$,
 \begin{align}
A^1 \vee A^2=  (E_1^1\cup E_1^2) \times (E_2^1\cup E_2^2) \times  (G^1\cup G^2), \notag\\
A^1 \wedge A^2=  (E_1^1\cap E_1^2)\times (E_2^1\cap E_2^2)\times (G^1\cap G^2).\notag
 \end{align}
 Notice that for each $X=T_1,T_2,K_*$, $(2^X, \subseteq, \cap,\cup)$ is a distributive lattice. Since $\mathcal{A}$ is the direct product of distributive lattices, it is also a distributive lattice.

\begin{lm}
$(\mathcal{A}, \preceq, \wedge, \vee)$  is a distributive lattice.
\end{lm}

Consider the implementability condition \eqref{eq:cut2} where the row index set is given by $\mathcal{A}$ and the column index set is given by $C=  (T_1\times K_*)\cup (T_2\times K_*)$. By the change of variables $h_i^k(t_i)=Q_i^k(t_i)\lambda_i(t_i)$, the implementability condition reduces to a linear system with a $\{-1,0,1\}$ constraint matrix. To see this, we define  $h_j: \mathcal{A}\to \{-1,0,1\}$, $j\in C$,  and  $\beta: \mathcal{A} \to \R$ as follows. For each row $A=E_1\times E_2\times G$,  let $\beta(A)=\lambda(E_1\times E_2^c)$, and let
\begin{center}
 $ h_j(A) = \left\{ \begin{array}{ll}
+1  &  \text{ if}\phantom{0}   j\in E_1\times G, \\
-1  &  \text{ if}\phantom{0}  j \in E_2\times G,\\
0   &  \text{ if}\phantom{0}  j\notin (E_1\times G)\cup (E_2\times G).
\end{array} \right.$ \notag  \\
\end{center}

 We can write the set $\mathcal{Q}^*$ as
\begin{equation}
\mathcal{Q}^{*}= \{a\in \R_+^{|C|}:  h(A)\cdot a\leq  \beta (A), \,\,\text{for all}\,\, A\in \mathcal{A}\}.  \notag
\end{equation}

The next characterization of $\mathcal{Q}^*$ depends on the notion of lattice polyhedra introduced by \cite{GH82}. Lattice polyhedra is a general framework for various combinatorial structures, such as polymatroids and the intersection of polymatroids. They are specified by a lattice structure on the underlying matrix satisfying certain submodularity constraints, where the matrices used in the description of these polyhedra are $\{-1, 0, 1\}$.   In Appendix B, we formally introduce the notion of (distributive) lattice polyhedra  and we obtain the following characterization.

\begin{prop}(i) The set $\mathcal{Q}^*$ is a distributive lattice polyhedron.  (ii) The set $\mathcal{Q}$  is an intersection of a pointed cone and a distributive lattice polyhedron. \end{prop}

It is known that distributive lattice polyhedra can be reduced to submodular flow polyhedra \cite{SC04}, for which efficient algorithms have been developed. This means that a mechanism design problem in our environment can be solved as tractably in reduced form when the implementability condition is considered.

\subsection{A Player-symmetric Description}

Border's condition is  player-symmetric, in the sense that the condition is invariant under permutations of players' labels. On the contrary, the implementability condition \eqref{eq:cut2} is not player-symmetric, i.e., for each inequality the coefficients are  $\{0,1\} $ for player 1 and $\{-1,0\}$  for player 2. This is because the directed network constructed  is not symmetric with respect to $T_1$ and $T_2$.  We show that condition \eqref{eq:4} and \eqref{eq:cut2} can be packed into a more compact description, which is player-symmetric.

Condition \eqref{eq:4} shows that $Q$ is in the solution set of a homogeneous system of equations. Multiplying both sides of condition \eqref{eq:cut2} by $2$ and subtracting  \eqref{eq:4} for all $k\in G$, we get the following player-symmetric condition: for all $G\subseteq K_*$,  $ E_1\subseteq T_1 $, and $E_2\subseteq T_2$,
\begin{equation}\label{eq:cut5}
\frac{1}{2}\sum_{i\in \{1,2\}}\sum_{k\in G}[\sum_{t_i\in E_i} Q^k_i(t_i)\lambda_i(t_i)-\sum_{t_i\in   E_i^c} Q^k_i(t_i)\lambda_i(t_i)]   \leq  \lambda(E_1\times E_2).
 \end{equation}

It can be shown that condition \eqref{eq:cut5} implies both \eqref{eq:4} and \eqref{eq:cut2}. Note that condition \eqref{eq:cut5} requires half-integral coefficients $\{-\frac{1}{2},0,\frac{1}{2}\}$ in the constraint matrix of reduced forms, in contrast to integral coefficients $\{0, 1\}$ in Border's condition.

\section{Concluding Remarks}

In this paper, we have studied the implementation of reduced form allocations in 2-person bargaining problems and  characterized the implementability condition. We derive a set of necessary conditions for the implementability and show that it is a complete description for the implementability. We also find that the implementability condition forms a lattice polyhedron and the feasible reduced forms has some new features compared to Border's condition. A by-product of our analysis concerns the transformation technique for multiflows on suspension graphs, which suggests other applications in mechanism design problems, including the implementation of expected allocations in random assignment problems. 

On the other hand, the characterization result in this paper is limited to some special type sets and the  combinatorial structure of the implementability condition becomes more complicated compared to the one-dimensional problems. In particular, the implementation condition does not form a polymatroid, and hence reduced form allocations may not be written as convex combinations of simple (hierarchical) allocations. Moreover, the characterization by the implementability condition requires a large amount of constraints.  A reduction of constraints similar to \cite{BO91} and \cite{CKM13} remains a challenging problem.  Another important issue would be finding polynomial time algorithms (in the number of types and alternatives) for optimization and implementation of interim allocation rules as  \cite{AFH19}. We hope that the implementation result in this paper will be useful to study these problems.

\begin{center}
{\large \textbf{Appendix A}}
\end{center}
In this appendix, we provide a detailed proof of Theorem 1.  Let $ \mathcal{I}$ be an implementation problem. For any $q$ and $Q$, we denote $x^k(t):= q^k(t)\lambda(t)$ and $h_i^k(t_i):= Q_i^k(t_i)\lambda_i(t_i)$.  As we will show below, the  change of variables reveals a clear combinatorial structure of the implementation system. We define a linear map $N: \R^T\to \R^{T_1+T_2}$ by $(N v)_i(t_i):= \sum_{t_{-i}\in T_{-i}} v(t) $ for any $v\in \R^T$, which corresponds to the node-edge incidence matrix of a complete bipartite graph on $T_1\cup T_2$.  Since $\lambda>0$, we reformulate the implementation system $\mathcal{I}$   by the following multiflow system $P_{\mathcal{I}}$, with $x$ being the multiflow variable:
 \begin{alignat}{3}
&   N x^k= h^k,     \,\,\text{for all}\,\,k\in K_*, \label{eq:B2}  \notag \\
& {\sum_{k\in K_*}} x^k+x^{k_0}= \lambda, \notag\\
& x^k\geq 0, x^{k_0}\geq 0. \notag
\end{alignat}

\begin{lm} (1)$P_{\mathcal{I}}$ has a feasible solution if and only if the transformed problem  $P(t_1^a)$ defined by
\begin{align}
  {\sum_{t_{-i}: t\in T}}  x^k(t)  &=h_i^k(t_i),    \,\, \text{for all}\,\,  (k,t_i)\in K\times (T_2\cup T_1\setminus\{t_1^a\}),  \notag\\
   {\sum_{k\in K}}  x^k(t) &=\lambda(t),  \,\, \text{for all}\,\,  t=(t_1^a,t_2)\in T,   \notag\\
x&\geq 0, \notag
\end{align}
 has a feasible solution.  (2) $P(t_1^a)$ is a network flow problem $(G(t_1^a), c, d)$ where 

i.  The transformed graph $G(t_1^a)=(V_1\cup V_2,E)$  is a bipartite graph:  

 The supply nodes $V_1$  consist of all $(k,t_2)\in K\times T_2$, and the  demand nodes   $V_2$  consist of  all $(k,t_1)\in  K\times  (T_1\setminus \{t_1^a\})$ and all $(t_1^a,t_2)\in T$.  

 The arc set $E$ consists of  arcs from each $(k,t_2)$ to   $(k',t_1)$ if $k'=k$, and arcs from each $(k,t_2)$ to $(t_1^a,t_2')$ if $t_2'=t_2$.

ii. The supply $c: V_1\to \R_+$ and the demand $d: V_2\to \R_+$ at each node is defined by the original system, i.e., each row in $P(t_1^a)$ is a row in $P_{\mathcal{I}}$.
\end{lm}

\begin{proof}[Proof of Lemma 7]  (1)   Notice that $(T_1\cup T_2, T)$ with $T_1=\{t_1^a,t_1^b\}$ is a suspension. By Lemma 5, setting $v_*=t_1^a$ in $P_{\mathcal{I}}$ leads to the transformed problem $P(t_1^a)$.   (2) By inspection,  each variable in $P(t_1^a)$ appears in exactly two constraints with opposite signs. Hence, the constraint matrix of   $P(t_1^a)$ is the node-arc incidence matrix of $G(t_1^a)$.  With supply vector $c$ and demand vector $d$  defined by the original system, $P(t_1^a)$  corresponds to a node-capacitated network flow problem $(G(t_1^a), c, d)$: $P(t_1^a)$ has a feasible solution $ x$ if and only if there exists a non-negative flow $f: E\to \R_+$ satisfying that the total flow leaving each $v\in V_1$ meets $c$ exactly and the total flow into each $v\in V_2$ meets $d$ exactly.  \end{proof}

To obtain the characterization in Theorem 1, we apply the following version of Hall's theorem. It provides a necessary and sufficient condition for the existence of a perfect matching in a bipartite graph with a non-unity capacity at each node.

 \begin{lm}\citep[p.15]{FU05} Let $G=(V_1\cup V_2, E)$ be a bipartite graph and let supply and demand vectors $c\in\R^{V_1}$ and  $d\in \R^{V_2}$ be given. There exists a non-negative flow $\varphi: E\to \R_+ $ satisfying that the total flow leaving each $v\in V_1$ meets  exactly the supply and the total flow into each $v\in V_2$ meets  exactly the demand if and only if
 \begin{equation}
\label{FU05}   d(U)\leq  c(\Gamma(U)),\,\,\,\text{for all} \,\,\, U\subseteq V_2,
\end{equation}
where $\Gamma(U)$ denotes the set of all neighbors of $U$ in $G$.
 \end{lm}

\begin{proof}[Proof of Theorem 1] By Lemma 7, $P(t_1^a)$ has a feasible solution if and only if for the network flow problem $(G(t_1^a),c,d)$, there exists a non-negative flow that meets the supply and demand at each node. By Lemma 8,  $P(t_1^a)$ has a feasible flow if and only if for all $U\subseteq V_2$,
 \begin{equation} \label{Z5}
 \sum_{(k,t_1)\in U } h_1^k(t_1)+\sum_{(t_1^a,t_2)\in U} \lambda(t_1^a,t_2) \leq \sum_{(k,t_2)\in \Gamma(U) } h_2^k(t_2).
  \end{equation}

Fix any $U\subseteq V_2$. We have  $U=(U_1\times \{t_1^b\})\cup (\{t_1^a\}\times U_2)$ for some $U_1\subseteq K$ and $U_2\subseteq T_2$. It follows that $\Gamma(U_1\times \{t_1^b\})\setminus (K\times U_2)= U_1\times (T_2\setminus U_2)$.  We can write condition \eqref{Z5} as
 \begin{equation} \label{Z6}
 \sum_{k\in U_1 } h_1^k(t_1^b)+\sum_{t_2\in U_2} \lambda(t_1^a,t_2) \leq \sum_{(k,t_2)\in K\times U_2 } h_2^k(t_2)+\sum_{(k,t_2)\in U_1\times (T_2\setminus U_2) } h_2^k(t_2).
 \end{equation}
Note that
 \begin{equation}
\sum_{(k,t_2)\in K\times U_2 } h_2^k(t_2)=\sum_{t\in T: t_2\in U_2} \lambda (t),\notag \end{equation}
which implies
\begin{align}\label{Z7}
\sum_{(k,t_2)\in K\times U_2 } h_2^k(t_2)-\sum_{t_2\in U_2} \lambda(t_1^a,t_2)=\lambda(\{t_1^b\}\times U_2).
\end{align}
By conditions \eqref{Z6} and  \eqref{Z7},  \eqref{Z5} can be written as
\begin{equation} \label{eq:T1a}
 \sum_{k\in U_1 } h_1^k(t_1^b)-\sum_{(k,t_2)\in U_1\times (T_2\setminus U_2)} h_2^k(t_2) \leq \lambda(\{t_1^b\}\times U_2).
 \end{equation}
There are two cases, depending on whether  $k_0\in U_1$ or not.

Suppose first that $k_0\notin U_1$. Then  \eqref{eq:T1a} reduces to the implementability inequality \eqref{eq:cut2},  with  $G=U_1$, $E_1=\{t_1^b\}$,  and $E_2=T_2\setminus U_2$.

Suppose next that $k_0\in U_1$. On the left hand side of \eqref{eq:T1a}  we take the summations for $k\in U_1\setminus\{k_0\}$ and $k=k_0$. Note that
\begin{equation}
 h_1^{k_0}(t_1^b)-\sum_{t_2\notin U_2} h_2^{k_0}(t_2)=\lambda_1(t_1^b)-\sum_{k\in K_*}h_1^k(t_1^b)-\sum_{t_2\notin U_2} \lambda_2(t_2)+\sum_{k\in K_*}\sum_{t_2\notin U_2} h_2^k(t_2). \notag
\end{equation}
The left hand side of \eqref{eq:T1a} is given by
\begin{equation}
- \sum_{k\in K_*\setminus  U_1}[h_1^k(t_1^b)-\sum_{t_2\notin U_2} h_2^k(t_2)]+\lambda_1(t_1^b)-\sum_{t_2\notin U_2} \lambda_2(t_2).\notag
 \end{equation}
Hence condition \eqref{eq:T1a} reduces to
\begin{equation}\label{eq:T1b}
- \sum_{k\in K_*\setminus  U_1}[h_1^k(t_1^b)-\sum_{t_2\notin U_2} h_2^k(t_2)]\leq \sum_{t_2\notin U_2}\lambda(t_1^a, t_2).
 \end{equation}
Since $h_1^k(t_1^b)=\sum_{t_2\in T_2}h_2^k(t_2)-h_1^k(t_1^a)$,   it follows that \eqref{eq:T1b} reduces to the implementability inequality \eqref{eq:cut2}, with  $G=K_*\setminus  U_1$, $E_1=\{t_1^a\}$,  and $E_2= U_2$.  \end{proof}

\begin{center}
{\large \textbf{Appendix B}}
\end{center}

Before introducing the lattice polyhedron,  we introduce some notions.  Let  $(L, \preceq, \wedge, \vee)$ be a finite lattice with partial order $\preceq$, meet $\wedge$,  and join $\vee$. We say $f: L\to \R$ is {\it submodular} if $f(a)+f(b)\geq  f(a \vee  b)+   f(a \wedge b)$ for all  $a, b\in L$,  {\it supermodular} if the inequality is reversed, and {\it modular} if equality holds.

\begin{df}\citep{GH82} Let  $(L, \preceq, \wedge, \vee)$ be a finite lattice. Let $r: L\to \R$, $f_j : L\to  \{0, \pm1\}^n$  for $j=1,...,n$, and let $c,d\in \{\R^n\cup\pm \infty\}$. The polyhedron
\begin{equation}
P=  \{x\in \R^n:   c\leq x\leq d,     f(a)\cdot x \leq r(a), \,\,\text{for all}\,\, a\in L\} \notag
\end{equation}
is called a {\it lattice polyhedron}, if  $r$  is submodular, and  for each $j=1,...,n$, the following three conditions hold: for any $a,b,c\in L$,

 (C1)  $|f_j(a)-f_j(b)|\leq 1$ if $a\prec b$,

 (C2)    $|f_j(a)-f_j(b)+f_j(c)|\leq 1$  if $a\prec b\prec c$,  and

 (C3) $f_j(a)+f_j(b)\leq  f_j(a \vee  b)+f_j(a \wedge b)    $.

 Moreover, a lattice polyhedron  $P$ is a {\it distributive lattice polyhedron}, if the lattice $(L, \preceq, \wedge, \vee)$ is distributive and (C3) is satisfied with equality, i.e., for all $j=1,...,n$, $f_j$ is modular.
 \end{df}

\begin{proof}[Proof of Proposition 1] By Lemma 6,   $(\mathcal{A}, \preceq, \wedge, \vee)$  is a  distributive lattice. We show that $ \mathcal{Q}^{*}$ defines a distributive lattice polyhedron with respect to $(\mathcal{A}, \preceq, \wedge, \vee)$.

We first show that   $\beta$ is submodular. To this end, let $(V, E)$ be a complete bipartite graph where $V=T_1\cup T_2$ consists of all nodes and  $E$ consists of all arcs from each $t_1\in T_1$ to each $t_2\in T_2$. Define a capacity function $c:E\to \R_+$ by $c(e)=\lambda(e)$. For any $U\subseteq V$, let $\delta^{out} (U)$ denote the arcs leaving $U$. Then, the cut function $g: 2^V\to \R$ defined by  $g(U)=c(\delta^{out} (U)) $ is submodular, with set union and intersection as the lattice operations.  For each $A=E_1\times E_2\times G$,   let $\pi(A)=E_1\cup E_2$.  By construction,   $\beta(A)=g(\pi(A))$,   $\beta(A^1 \vee  A^2)=g(\pi(A^1)\cup \pi(A^2))$,  and $\beta(A^1 \wedge A^2)=g(\pi(A^1)\cap \pi(A^2))$.  Since $g$ is submodular, we have for any $A^1,A^2\in \mathcal{A}$,
\begin{align}
&\beta(A^1)+\beta(A^2)-\beta(A^1 \vee  A^2)-   \beta(A^1 \wedge A^2)  \notag\\
&=g(\pi(A^1)) +g(\pi(A^2))-g(\pi(A^1) \cup \pi(A^2))- g(\pi(A^1) \cap \pi(A^2))\geq0.  \notag
\end{align}

We then show that for each $j\in C$, $h_j$ satisfies  (C1)-(C3) and (C3) holds with equality. We verify (C1) and the other conditions can be shown similarly.  Let $A^1\prec A^2$. Pick any $j\in C$, then either  $j=(t_1,k)$ or $j=(t_2,k)$. We consider $j=(t_1,k)$ and a similar analysis applies to $j=(t_2,k)$.  There are three subcases: (a) $j\in E_1^1\times G^1$, (b)  $j\in (E_1^2\times G^2)\setminus (E_1^1\times G^1)$, and (c) $j\notin E_1^2\times G^2$. In each of the cases, $|h_j(A^1)-h_j(A^2)|\leq 1$.  \end{proof}

\end{document}